\documentclass[journal,12pt,draftclsnofoot,onecolumn]{IEEEtran}
\usepackage{amsmath}
\usepackage{graphicx}
\usepackage{epstopdf}
\usepackage{amssymb}
\usepackage{amsthm}
\usepackage{caption}
\usepackage{cite}
\usepackage{subfigure}
\usepackage{enumerate}
\usepackage{textcomp}
\usepackage{tikz}
\usepackage{subfigure}
\usetikzlibrary{patterns}
\usetikzlibrary{shapes,arrows}
\usepackage{verbatim}
\usetikzlibrary{positioning}
\usetikzlibrary{shapes.geometric}
\usetikzlibrary{shapes.misc}
\begin{document}
\title{Impact of Correlation between Interferers on Coverage Probability and rate in Cellular Systems}
\author{Suman Kumar \hspace*{0.2in} Sheetal Kalyani  \\
\hspace{0in} Dept. of Electrical Engineering \\
  \hspace{-0in}   IIT Madras,  Chennai 600036, India   \\ 
{\tt \{ee10d040,skalyani\}@ee.iitm.ac.in}\\
}

\maketitle
\begin{abstract} 
When the user channel experiences Nakagami-m fading, the coverage probability expressions are theoretically compared for the following cases: $(i).$ The $N$ interferers are independent $\eta$-$\mu$ random variables (RVs). $(ii).$ The  $N$ interferers are correlated $\eta$-$\mu$ RVs. It is  analytically shown that  the coverage probability in the presence of correlated interferers is greater than or equal to the coverage probability in the presence of  independent interferers  when the  shape parameter of the channel between the user and its base station (BS)  is not greater than one. Further, rate is compared for the following cases: $(i).$  The user channel experiences $\eta$-$\mu$ RV and the $N$ interferers are independent $\eta$-$\mu$ RVs. $(ii).$ The  $N$ interferers are correlated $\eta$-$\mu$ RVs. It is   analytically shown that the rate in the presence of correlated interferers is greater than  or equal to the rate in the presence of independent interferers.  Simulation results are provided  and these match with the obtained theoretical results. The utility of our results are also discussed. 
\end{abstract}
\section{Introduction}

Typically, in practical scenarios correlation exists among the interferers, as evidenced by experimental results reported in  \cite{correlation,1105, 82767, 104090, 490700, 944855}.    For example, in cellular networks when two base stations (BSs) from adjacent sectors act as interferers, the interferers are correlated and it is mandated that while performing system level simulation, this correlation be explicitly introduced in the system \cite{3gpp}. We refer the  reader to  \cite{5590312}   for a structured synthesis of the existing literature on correlation among large scale fading.  Considering the impact of correlation in the large scale  shadowing component and the small scale  multipath component is also an essential step towards modeling the channel.  The decorrelation distance in multipath components is lower when compared to shadowing components since shadowing is related to terrain configuration and/or large obstacles between transmitter and receiver \cite{991146}. Having said this, there is a need to analyse the performance of cellular system in the presence of correlation among interferers.

Coverage probability\footnote{It is a probability that a user can achieve a target Signal-to-Interference-plus-noise-Ratio (SINR) $T$, and outage probability is the complement of  coverage probability.} and rate are  important metric  for performance evaluation of cellular systems. Coverage probability in the presence of interferers has been derived in  \cite{6171806, paris2013outage, 6661325, 6502746, 6957529, 7130668} and references there in, and rate in the presence of interferers has been studied in \cite{6502746, 6957529, 7130668}  and references therein for the case of $\eta$-$\mu$ fading. Moreover, the correlation among interferers is assumed in \cite{7130668}. Fractional frequency reuse and soft frequency reuse have been compared in the presence of correlation among interferers \cite{7015633}. The impact of correlation among interferers on symbol error rate performance has been analysed in \cite{7173002}. However, to the best of our knowledge, no prior work in open literature has  analytically compared the coverage probability and rate for generalized fading when interferers are independent with the coverage probability and rate when interferers are correlated. In this paper, we compare coverage probability and rate using majorization theory and stochastic ordering theory, respectively. Majorization theory is an important theory for comparison of two vectors in terms of the dispersion of their components. It  has  been extensively used in various problems in information theory \cite{259659,992785, 1278662, 1056173, 1661855, 6409457}. In particular, using the results of majorization theory, a new analysis on Tunstall algorithm is provided in \cite{1661855} and the bounds of  the average length of the Huffman source code in the presence of limited knowledge of the source symbol probability distribution is provided in \cite{1278662}. Our analysis of comparison of coverage probability shows further evidence of the relevance of majorization theory. Stochastic ordering theory is used extensively for comparison of  random variables (RVs). Recently, it has been used in comparing various metrics in wireless communication \cite{6042309, 6423920,7031955, 6662523}. Our analysis of comparison of rate shows another evidence of the relevance of stochastic order theory in wireless communication.

In this work,  we compare the  coverage probability when user experience Nakagami-m fading and interferers experience $\eta$-$\mu$ fading. In other words, we compare the coverage probability when the interferers are independent   with the coverage probability when the interferers are positively correlated\footnote{If $cov(X_i, X_j)\geq 0$ then $X_i$ and $X_j$  are positively correlated RVs, where $cov(X_i, X_j)$ denotes the covariance between $X_i$ and $X_j$ \cite{403769}.} using majorization theory.   It is analytically shown that the coverage  probability in presence of correlated interferers is higher than the coverage probability when the interferers  are independent, when the user channel's shape parameter is lesser than or equal to one.   We also show that when the user channel's shape parameter is greater than one, one cannot say whether coverage 
probability is higher or lower for the correlated case when compared to the independent case, and in some cases coverage probability is higher while in other cases it is lower.

We then analytically compare the rate when the interferers are independent  with the rate when the interferers are correlated using  stochastic ordering theory.   It is shown that the rate in the presence of positively correlated interferers is higher than  the rate in the presence of independent interferers when both user channel and interferers experience $\eta$-$\mu$ fading. Our results show that correlation among interferers is beneficial for the desired user. We briefly discuss how the desired user can exploit this correlation among the interferers to improve its rate. Multi-user multiple input multiple output (MU-MIMO) system is also considered and it is shown that the impact of correlation is significant on MU-MIMO system. We have also carried out extensive simulations for both the independent interferers case and the correlated interferers case and some of these results are reported in the Simulation section.
In all the cases, the simulation results match with our theoretical results.

\section{System model}
We consider a homogeneous macrocell network with hexagonal structure with radius $R$ as shown in Fig. \ref{fig:hexagonal}. The Signal-to-Interference-Ratio  (SIR)  of a user located at $r$ meters from the BS is given by
\begin{equation}
\text{SIR}=\eta(r)= \frac{Pgr^{-\alpha}}{\sum\limits_{i\in\psi}Ph_id_i^{-\alpha}}= \frac{gr^{-\alpha}}{\sum\limits_{i\in\psi}h_id_i^{-\alpha}}=\frac{S}{I} \label{sir}
\end{equation}
where $\psi$ denotes the set of interfering BSs and  $N=|\phi|$ denotes the cardinality of the set $\phi$. The transmit power of a BS is denoted by $P$.  A standard path loss model $r^{-\alpha}$ is considered, where $\alpha\geq 2$ is the path loss exponent.  The distance between user to tagged BS (own BS) and the $i$th interfering BS is denoted by $r$ and $d_i$, respectively. The user channel's power  and the channel power between $i^{th}$ interfering BS and user are $\eta$-$\mu$ power RVs. The probability density function (pdf) $f_{g_{{\eta-\mu}}(x)}$ of the $\eta$-$\mu$ power RV $g$ is given by \cite[Eq. (26)]{4231253},
\begin{equation}
f_{g_{{\eta,\mu}}(x)}=\frac{2\sqrt{\pi}\mu^{\mu+ \frac{1}{2}}h^{\mu} x^{\mu- \frac{1}{2}}}{\Gamma(\mu)H^{\mu- \frac{1}{2}}}e^{-2\mu h x}I_{\mu- \frac{1}{2}}(2\mu H x)\label{eta_mu2}
\end{equation}
where $\mu$ is shape parameter. Parameters $H$ and $h$  are given by
\begin{equation}
H=\frac{\eta^{-1}-\eta}{4}, \text{ and  } h=\frac{2+\eta^{-1}+\eta}{4}.
\end{equation}
where $0<\eta< \infty $ is the power ratio of the in-phase and quadrature 	component of the fading signal in each multipath cluster. 
The parameters of pdf of $h_i$ are $\eta_i$ and $\mu_i$ corresponding to $\eta$-$\mu$ power RV. The gamma RV is a special case of $\eta$-$\mu$  power RV with $\eta=1$ and $\mu=\frac{m}{2}$ and the pdf $f_g(x)$  of the gamma RV $g$  is given by
\begin{equation}
f_g(x)=m^m e^{-mx}\frac{x^{m-1}}{\Gamma(m)}\label{gamma}
\end{equation}
where $m$ and $\frac{1}{m}$ are the shape parameter and scale  parameter, respectively, and $\Gamma(.)$ denotes the gamma function. When one assumes that there is correlation among interferers then $h_i$ and $h_j$ are correlated $\forall\text{ } i \text{ and }j$.
\begin{figure}[ht]
\centering
\begin{tikzpicture}
\node[pattern=north west lines, regular polygon, regular polygon sides=6,minimum width= 1  cm, draw] at (5,5) {};
\node at (5,5){\scriptsize $0$};
\node[regular polygon, regular polygon sides=6,minimum width=1 cm, draw] at (5+1.5*0.5,5-0.866*0.5) {};
\node at (5+1.5*0.5,5-0.866*0.5){\scriptsize $6$};
\node[ regular polygon, regular polygon sides=6,minimum width=1 cm, draw] at (5+1.5*0.5,5+0.866*0.5) {};
\node at (5+1.5*0.5,5+0.866*0.5){\scriptsize $1$};
\node[ regular polygon, regular polygon sides=6,minimum width=1 cm, draw] at (5-1.5*0.5,5+0.866*0.5) {};
\node at (5-1.5*0.5,5+0.866*0.5){\scriptsize $3$};
\node[ regular polygon, regular polygon sides=6,minimum width=1 cm, draw] at (5-1.5*0.5,5-0.866*0.5) {};
\node at (5-1.5*0.5,5-0.866*0.5){\scriptsize $4$};
\node[ regular polygon, regular polygon sides=6,minimum width=1 cm, draw] at (5,5-0.866) {};
\node at (5,5-0.866){\scriptsize $5$};
\node[ regular polygon, regular polygon sides=6,minimum width=1 cm, draw] at (5,5+0.866) {};
\node at (5,5+0.866){\scriptsize $2$};
\node[ regular polygon, regular polygon sides=6,minimum width=1 cm, draw] at (5,5+0.866*2) {};
\node at (5,5+0.866*2){\scriptsize $9$};
\node[ regular polygon, regular polygon sides=6,minimum width=1 cm, draw] at (5,5-0.866*2) {};
\node at (5,5-0.866*2){\scriptsize $15$};
\node[ regular polygon, regular polygon sides=6,minimum width=1 cm, draw] at (5+3*0.5,5+0.866*1) {};
\node at (5+3*0.5,5+0.866*1){\scriptsize $7$};
\node[ regular polygon, regular polygon sides=6,minimum width=1 cm, draw] at (5+3*0.5,5-0.866*1) {};
\node at (5+3*0.5,5-0.866*1){\scriptsize $17$};
\node[ regular polygon, regular polygon sides=6,minimum width=1 cm, draw] at (5-3*0.5,5-0.866*1) {};
\node at (5-3*0.5,5-0.866*1){\scriptsize $13$};
\node[ regular polygon, regular polygon sides=6,minimum width=1 cm, draw] at (5-3*0.5,5+0.866*1) {};
\node at (5-3*0.5,5+0.866*1){\scriptsize $11$};
\node[ regular polygon, regular polygon sides=6,minimum width=1 cm, draw] at (5+1.5,5) {};
\node at (5+1.5,5){\scriptsize $18$};
\node[ regular polygon, regular polygon sides=6,minimum width=1 cm, draw] at (5-1.5,5) {};
\node at (5-1.5,5){\scriptsize $12$};
\node[ regular polygon, regular polygon sides=6,minimum width=1 cm, draw] at (5+1.5*0.5,5-0.866*1.5) {};
\node at(5+1.5*0.5,5-0.866*1.5){\scriptsize $16$};
\node[ regular polygon, regular polygon sides=6,minimum width=1 cm, draw] at (5-1.5*0.5,5+0.866*1.5) {};
\node at(5-1.5*0.5,5+0.866*1.5){\scriptsize $10$};
\node[ regular polygon, regular polygon sides=6,minimum width=1 cm, draw] at (5-1.5*0.5,5-0.866*1.5) {};
\node at(5-1.5*0.5,5-0.866*1.5){\scriptsize $14$};
\node[ regular polygon, regular polygon sides=6,minimum width=1 cm, draw] at (5+1.5*0.5,5+0.866*1.5) {};
\node at(5+1.5*0.5,5+0.866*1.5){\scriptsize $8$};
\draw[<->] (5,5)--(5+1.5*0.5,5+0.866*0.5) node[pos=0.75,sloped,above] {\tiny$2R$};
\end{tikzpicture}
            \caption{ Macrocell network with hexagonal tessellation having inter cell site distance $2R$}
           
             \label{fig:hexagonal}
        \end{figure}
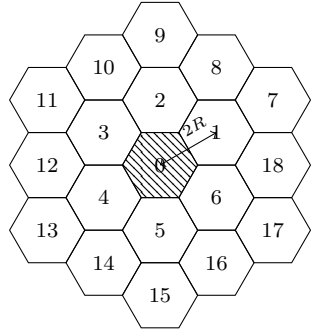

The coverage probability expression when signal of interest (SoI) experience Nakagami-m fading and interferers experience independent $\eta$-$\mu$ fading with equal $\mu$, i.e.,  $\mu_i=\mu_c \forall i$ is given by \cite{7130668}
\begin{equation*}
C_{p,\eta}(r)=\frac{\Gamma(2N\mu_c+m)}{\Gamma\left(2N\mu_c+1\right)}\frac{1}{\Gamma{(m)}} \prod\limits_{i=1}^{2N}\left(\frac{1}{Tr^{\alpha}m \lambda_i +1}\right)^{\mu_c} \times
\end{equation*}
\begin{equation}
 F_D^{(2N)}\left[1-m,\mu_c,\cdots,\mu_c; 2N\mu_c+1;\frac{1}{Tr^{\alpha}  \lambda_1m+1},\cdots,\frac{1}{Tr^{\alpha}\lambda_{2N} m+1}\right] \label{eq:cov_eta_mu_system} 
\end{equation}
\begin{equation}
\text{ where } \lambda_{2i-1}=\frac{d_i^{-\alpha}}{\mu_c(1+\eta_i^{-1})} \text{ and } \lambda_{2i}=\frac{d_i^{-\alpha}}{\mu_c(1+\eta_i)}. \label{lambda_system}
\end{equation}
Here $F_D^{(N)}\left[a,b_1, \cdots, b_N;c;x_1,\cdots, x_N\right]$ is the  Lauricella's function of the fourth kind \cite{exton1976multiple}. 
The coverage probability expression when SoI experience Nakagami-m fading and interferers experience correlated $\eta$-$\mu$ fading with equal $\mu$, i.e.,  $\mu_i=\mu_c \forall i$ is given by \cite{7130668}
\begin{equation*}
\hat{C}_{p,\eta}(r)=\frac{\Gamma(2N\mu_c+m)}{\Gamma\left(2N\mu_c+1\right)}\frac{1}{\Gamma{(m)}} \prod\limits_{i=1}^{2N}\left(\frac{1}{Tr^{\alpha}m \hat{\lambda}_i +1}\right)^{\mu_c} \times
\end{equation*}
\begin{equation}
 F_D^{(2N)}\left[1-m,\mu_c,\cdots,\mu_c; 2N\mu_c+1;\frac{1}{Tr^{\alpha}  \hat{\lambda}_1m+1},\cdots,\frac{1}{Tr^{\alpha}\hat{\lambda}_{2N} m+1}\right] \label{eq:cov_eta_mu_corr_system} 
\end{equation}
Here $\hat{\lambda}_i$s are the eigenvalues of  $\mathbf{A_{\eta}=D_{\eta}C_{\eta}}$. $\mathbf{D}_{\eta}$ is the diagonal matrix with entries $\lambda_i$ given in \eqref{lambda_system} and $\mathbf{C}_{\eta}$ is the s.p.d. $2N\times 2N$ matrix as given by
\begin{equation}
\mathbf{C}_{\eta}=\left[ \begin{array}{ccccc}
 1 & 0& \sqrt{\rho_{13}} &...&0   \\ 0& 1 &0&...&\sqrt{\rho_{22N}} \\ \cdots &  \cdots &\cdots &\ddots & \cdots \\ 0&\sqrt{\rho_{2N2}} &\cdots &\cdots &1  \end{array} \right],\label{correlation1_system}
\end{equation}
here $\rho_{ij}=0$, when $i+j=2n+1$, where $n$ is an integer. 

Now, we want to compare \eqref{eq:cov_eta_mu_system} with   \eqref{eq:cov_eta_mu_corr_system}. However both involve $N$- fold infinite series making any comparison fairly difficult.  In order to elegantly compare these two expressions we make use of Majorization theory.

\section{Preliminaries: Majorization and  Stochastic Ordering Theory}
In this section, we recall the basic notions of majorization and  stochastic ordering theory, which are applicable to our context.  Note that the majorization theory is used to compare deterministic vectors, whereas stochastic order is applicable on RVs \cite{6042309}. We refer the reader to  \cite{marshall2011inequalities} and \cite{shaked2007stochastic} as excellent references for the majorization theory and stochastic order theory, respectively.
\subsection{Majorization theory}
\newtheorem{definition}{Definition}
\begin{definition}
\label{Hadamard}
\textbf{Let $A$ and $B$ be $m\times n$ matrices with entries in $C$. The Hadamard product of $A$ and $B$ is defined by $[A \circ B]_{ij}=[A]_{ij}[B]_{ij}$ for all $1\leq i\leq m$, $1\leq i\leq n$.}
\end{definition}

\begin{definition}
Let vector $\boldsymbol{a}=[a_1,\cdots a_n]$ and vector $\boldsymbol{b}=[b_1,\cdots b_n]$ with  $a_1\leq, \cdots,\leq  a_n$ and  $b_1\leq, \cdots,\leq  b_n$ then vector $\boldsymbol{b}$ is majorized by vector $\boldsymbol{a}$, denoted by $\boldsymbol{a}\succ \boldsymbol{b}$, if and only if 
\begin{equation}
\sum\limits_{i=1}^{k}b_i\geq \sum\limits_{i=1}^{k}a_i, \text{ } k=1,\cdots, n-1, \text{ and } \sum\limits_{i=1}^{n}b_i= \sum\limits_{i=1}^{n}a_i. 
\end{equation}
\end{definition}
\newtheorem{lemma}{Lemma} 
\newtheorem{theorem}{Theorem}
We will briefly state few well known results from  majorization theory which we use in our analysis.
\begin{theorem}
\label{majorization_theory1}
\textbf{ Let $A$ and B be positive semidefinite matrices of size $n$. Let $\lambda_1,\cdots, \lambda_n $ be the eigenvalues of $A \circ B$ and let $\hat{\lambda}_1, \cdots, \hat{\lambda}_n$ be the eigenvalues of $AB$. Then
\begin{equation}
\prod\limits_{i=k}^{n}\lambda_i \geq \prod\limits_{i=k}^{n} \hat{\lambda}_i, k=1,2,\cdots,n
\end{equation}}

\end{theorem}
\newtheorem{proposition}{Proposition}
\begin{proposition}
\label{majorization_thory1}
If function $\phi $ is symmetric and convex, then $\phi$ is Schur-convex function. Consequently, $\boldsymbol{x} \succ \boldsymbol {y}$ implies $\phi(\boldsymbol{x})\geq \phi(\boldsymbol{y})$.
\end{proposition}
\begin{proof}
For the details of this proof please refer to \cite[P. 97, C.2.]{marshall2011inequalities}.
\end{proof}
\begin{theorem}
\label{schur_convex}
If $\phi_i$ is Schur-convex, $i=1,\cdots k$, and $\phi_i(x)\geq 0$ for all $i$ and $x$, then 
\begin{equation}
\psi(x)=\prod\limits_{i=1}^{k}\phi_i(x)
\end{equation}
is Schur-convex.
\end{theorem}
\begin{proof}
For the details of this proof please refer to \cite[P. 97, B.1.d.]{marshall2011inequalities}.
\end{proof}
\subsection{Stochastic order theory}
In this subsection, our focus will be on convex order. 
\begin{definition}
If $X$ and $Y$ are two RVs such that 
\begin{equation}
E[\phi(X)]\leq E[\phi(Y)]\label{convex_order}
\end{equation}
for all convex function $\phi: \mathbb{R}\rightarrow \mathbb{R}$,  provided the expectation exist, then $X$ is said to be smaller than $Y$ in the convex order, denoted by $X \leq_{cx} Y$.
\end{definition}
Note that if \eqref{convex_order} holds then  $Y$ is \textbf{more variable} than $X$ \cite{shaked2007stochastic}. We now briefly state the theorems in convex order theory  that is relevant to this work.
\begin{theorem}
\label{stochastic_order}
Let $X_{1}, X_{2},\cdots, X_{N}$ be exchangeable RVs. Let $\boldsymbol{a}=(a_1,a_2,\cdots,a_N)$ and $\boldsymbol{b}=(b_1,b_2,\cdots,b_N)$ be two vectors of constants. If $\boldsymbol{a}\prec \boldsymbol{b}$, then
\begin{equation}
\sum\limits_{i=1}^{N}a_iX_i\leq_{cx}\sum\limits_{i=1}^{N}b_iX_i.
\end{equation}
\end{theorem}
\begin{proof}
The details of the proof is given in \cite[Theorem 3.A.35]{shaked2007stochastic}.
\end{proof}

\begin{theorem}
\label{stochastic_order1}
If $X\leq_{cx} Y$ and $f(.)$ is convex, then $E[f(X)]\leq E[f(Y)]$.
\end{theorem}
\begin{proof}
The details of the proof is given in \cite[Theorem 7.6.2]{kaas2001modern}.
\end{proof}

\section{Comparison of Coverage Probability} 
In this section, we first  compare the coverage probability in the independent case and correlated case, and analytically quantify the impact of correlation when both  SoI and interferers experience Nakagami-m fading. Then we analyse the impact of correlation on coverage probability for the scenario where SoI experiences Nakagami-m fading and interferers experience $\eta$-$\mu$ fading. The coverage probability expression when SoI experiences Nakagami-m fading and interferers  experience independent Nakagami-m fading  is given by \cite{7130668}
\begin{equation*}
C_p(r)=\frac{\Gamma(\sum\limits_{i=1}^{N}m_i+m)}{\Gamma\left(\sum\limits_{i=1}^{N}m_i+1\right)}\frac{1}{\Gamma{(m)}} \prod\limits_{i=1}^{N}\left(\frac{1}{Tr^{\alpha}m\lambda_i+1}\right)^{m_i} \times
\end{equation*}
\begin{equation}
 F_D^{(N)}\left[1-m,m_1,\cdots,m_{N}; \sum\limits_{i=1}^{N}m_i+1;\frac{1}{Tr^{\alpha} m\lambda_1+1},\cdots,\frac{1}{Tr^{\alpha}m\lambda_N+1}\right] \label{eq:cov7} \text{ where } \lambda_i=\frac{d_i^{-\alpha}}{m_i}
\end{equation}
Here $m_i$ is the shape parameter of the $i^{th}$ interferer. Note that the coverage probability expression for the correlated case is derived in \cite{7130668} when the interferers shape parameter are all equal, i.e., $m_i=m_c \forall i$.   The coverage probability expression when  SoI experiences Nakagami-m fading and interferers  experience correlated Nakagami-m fading is given by \cite{7130668}
\begin{equation*}
\hat{C}_p(r)=\frac{\Gamma(Nm_c+m)}{\Gamma\left(Nm_c+1\right)}\frac{1}{\Gamma{(m)}} \prod\limits_{i=1}^{N}\left(\frac{1}{Tr^{\alpha}m\hat{\lambda}_i +1}\right)^{m_c} \times
\end{equation*}
\begin{equation}
 F_D^{(N)}\left[1-m,m_c,\cdots,m_{c}; Nm_c+1;\frac{1}{Tr^{\alpha}m\hat{\lambda}_1+1},\cdots,\frac{1}{Tr^{\alpha}m\hat{\lambda}_N+1}\right] \label{eq:cov8}
\end{equation}
here $\hat{\lambda}_i$s are the eigenvalues of the matrix $\mathbf{A=DC}$, where $\mathbf{D}$ is the diagonal matrix with entries $\lambda_i$ and $\mathbf{C}$ is the symmetric positive definite (s.p.d.) $N\times N$ matrix defined by
\begin{equation}
\mathbf{C}=\left[ \begin{array}{cccc}
 1 & \sqrt{\rho_{12}} &...&\sqrt{\rho_{1N}}   \\ \sqrt{\rho_{21}} &1&...&\sqrt{\rho_{2N}} \\  \cdots &\cdots &\ddots & \cdots \\ \sqrt{\rho_{N1}} &\cdots &\cdots &1  \end{array} \right],\label{correlation}
\end{equation}
where $\rho_{ij}$ denotes the correlation coefficient between $h_i$ and $h_j$, and is given by,
\begin{equation}
\rho_{ij}=\rho_{ji}=\frac{cov(h_i,h_j)}{\sqrt{\text{var}(h_i)\text{var}(h_j)}}, 0\leq\rho_{ij}\leq 1, i,j=1,2, \cdots, N.
\end{equation}
${cov(h_i,h_j)}$   denotes the covariance between  $h_i$ and $h_j$. $\mathbf{C}$ is the s.p.d. $ N\times N$ matrix given in \eqref{correlation}.
 Since $\mathbf{A=DC}$, it is given by
\begin{equation}
 \mathbf{A}=\left[ \begin{array}{cccc}
  \lambda_1 & \lambda_1\sqrt{\rho_{12}} &\cdots&\lambda_1\sqrt{\rho_{1N}}   \\ \lambda_2\sqrt{\rho_{21}} &\lambda_2&\cdots&\lambda_2\sqrt{\rho_{2N}} \\ \vdots & \vdots&\ddots&\vdots \\ \lambda_N\sqrt{\rho_{N1}} &\cdots&\cdots&\lambda_N  \end{array} \right]. \label{A}
\end{equation}

Note that the coverage probability expression for the correlated case is derived when the interferers shape parameter are all equal and hence for a fair comparison we consider equal shape parameter for the independent case also, i.e., $m_i=m_c \text{ }\forall i $. Therefore, the coverage probability when interferers are independent and $m_i=m_c \text{ }\forall i $ is given by
\begin{equation*}
C_p(r)=\frac{\Gamma(Nm_c+m)}{\Gamma\left(Nm_c+1\right)}\frac{1}{\Gamma{(m)}} \prod\limits_{i=1}^{N}\left(\frac{1}{Tr^{\alpha}m \lambda_i +1}\right)^{m_c} \times
\end{equation*}
\begin{equation}
 F_D^{(N)}\left[1-m,m_c,\cdots,m_{c}; Nm_c+1;\frac{1}{Tr^{\alpha}  \lambda_1m+1},\cdots,\frac{1}{Tr^{\alpha}\lambda_N m+1}\right] \label{eq:cov9}
\end{equation}
Our goal is to compare the coverage probability in the presence of independent interferers, i.e., $C_p(r)$ given in \eqref{eq:cov9} and the coverage probability in the presence of correlated interferers, i.e., $\hat{C}_p(r)$ given in \eqref{eq:cov8}.
We first start with the special case when user channel's fading is Rayleigh fading (i.e, $m=1$) and interferers experience Nakagami-m fading. When $m=1$, the coverage probability given in \eqref{eq:cov9} reduces to
\begin{equation}
C_p(r)= \prod\limits_{i=1}^{N}\left(\frac{1}{Tr^{\alpha}\lambda_i +1}\right)^{m_c} F_D^{(N)}\left[0,m_c,\cdots,m_{c}; Nm_c+1;\frac{1}{Tr^{\alpha}\lambda_1+1},\cdots,\frac{1}{Tr^{\alpha}\lambda_N+1}\right] \label{eq:cov1}
\end{equation}
A series expression for $F_D^{(N)}( .)$ involving N-fold infinite sums is given by
\begin{equation} 
F_D^{(N)}[a,b_1,\cdots, b_N;c;x_1,\cdots, x_N]=\sum\limits_{i_1\cdots i_N=0}^{\infty}\frac{(a)_{i_1+\cdots+i_N}(b_1)_{i_1}\cdots(b_N)_{i_N}}{(c)_{i_1+\cdots+i_N}}\frac{x_1^{i_1}}{i_1!}\cdots\frac{x_N^{i_N}}{i_N!},\label{eq:lauricella1}
\end{equation}
\begin{equation*}
\max\{|x_1|,\cdots|x_N|\}<1,
\end{equation*}
where, $(a)_n$ denotes the Pochhammer symbol which is defined as $(a)_n=\frac{\Gamma(a+n)}{\Gamma(a)}$. With the help of series expression  and  using the fact   that $(0)_0=1$ and $(0)_k=0 \text{ }\forall \text{ } k\geq 1$,  the coverage probability given in \eqref{eq:cov1} can be reduced to 
\begin{equation}
C_p(r)= \prod\limits_{i=1}^{N}\left(\frac{1}{Tr^{\alpha}\lambda_i +1}\right)^{m_c}
\end{equation} 
Similarly, the coverage probability in correlated case when SoI experiences Rayleigh fading $\hat{C}_p(r)$ is given by
\begin{equation}
\hat{C}_p(r)= \prod\limits_{i=1}^{N}\left(\frac{1}{Tr^{\alpha}\hat{\lambda}_i +1}\right)^{m_c}
\end{equation}
We now state and prove the following theorem for the case where the SoI experiences Rayleigh fading  and interferers experience  Nakagami-m fading and then generalize it to the  case where user also experiences Nakagami-m fading. 
 
\begin{theorem}
\label{main_thoerem}
The coverage probability in correlated case is higher than that of the independent case, when user's  channel undergoes Rayleigh fading, i.e.,
\begin{equation}
\prod\limits_{i=1}^{N}\left(\frac{1}{{1+k\hat{\lambda}_i}}\right)^{m_c}\geq \prod\limits_{i=1}^{N}\left(\frac{1}{{1+k\lambda_i}}\right)^{m_c}
\end{equation}
where $\hat{\lambda}_i$s  are the eigenvalues of  matrix $\mathbf{A}$ and $\lambda_i$s are the diagonal elements of  matrix $\mathbf{A}$ given in \eqref{A}, and $k=Tr^{\alpha} $ is a non negative constant.
\end{theorem}
\begin{proof}
\textbf{ Firstly, we show that  $\boldsymbol{\hat{\lambda}} \succ \boldsymbol{\lambda}$ where $\boldsymbol{\hat{\lambda}}=[\hat{\lambda}_1, \cdots, \hat{\lambda}_n]$ and $\boldsymbol{\lambda}=[{\lambda_1}, \cdots, {\lambda_n}]$. Note that the Hadamard product of $\mathbf{D}$ and $\mathbf{C}$, i.e., $\mathbf{D} \circ \mathbf{C}$  is defined as\footnote{Definition of Hadamard product is given in Definition \ref{Hadamard}} }
\begin{equation}
 \mathbf{D} \circ \mathbf{C}=\left[ \begin{array}{cccc}
  \lambda_1 & 0 &\cdots&0   \\ 0 &\lambda_2&\cdots&0 \\ \vdots & \vdots&\ddots&\vdots \\ 0 &\cdots&\cdots&\lambda_N  \end{array} \right]. \label{A_circ}
\end{equation}
\textbf{Observe that $\mathbf{D} \circ \mathbf{C}$ is a diagonal matrix, the eigenvalues of the $\mathbf{D} \circ \mathbf{C}$ are $\lambda_i$s. Since the eigenvalues of  matrix $\mathbf{A}=\mathbf{D}\mathbf{C}$ are $\hat{\lambda_i}$s and $\mathbf{D}$ and $\mathbf{C}$ are the positive semi-definite matrices,  hence  from Theorem \ref{majorization_theory1} (given in Section II),}
\begin{equation}
\prod\limits_{i=k}^{N}\lambda_i \geq \prod\limits_{i=k}^{N} \hat{\lambda}_i, k=1,2,\cdots,N \label{weak1}
\end{equation}
\textbf{Note that $\hat{\lambda_i}$s are the eigenvalues and $\lambda_i$s are the diagonal elements of a symmetric matrix $\mathbf{A}=\mathbf{D}\mathbf{C}$, Hence}
\begin{equation}
\sum\limits_{i=1}^{N}\lambda_i=\sum\limits_{i=1}^{N}\hat{\lambda_i} \label{weak2}
\end{equation}
\textbf{ If conditions given in \eqref{weak1} and \eqref{weak2} are satisfied then $\boldsymbol{\hat{\lambda}} \succ \boldsymbol{\lambda}$ \cite{marshall2011inequalities} .} Now if it can be shown that $\prod\limits_{i=1}^{N}\left(\frac{1}{{1+k\lambda_i}}\right)^{m_c}$ $\left(\text { and }\prod\limits_{i=1}^{N}\left(\frac{1}{{1+k\hat{\lambda}_i}}\right)^{m_c}\right)$ is a Schur-convex function then by a simple application of Proposition \ref{majorization_thory1} (given in Section II) it is evident that $\prod\limits_{i=1}^{N}\left(\frac{1}{{1+k\hat{\lambda}_i}}\right)^{m_c}\geq \prod\limits_{i=1}^{N}\left(\frac{1}{{1+k\lambda_i}}\right)^{m_c}$. To prove that $\prod\limits_{i=1}^{N}\left(\frac{1}{{1+kx_i}}\right)^{m_c}$ is a Schur convex function we need to show that it is a symmetric and convex function \cite{marshall2011inequalities}. 

It is apparent that the function $\prod\limits_{i=1}^{N}\left(\frac{1}{{1+kx_i}}\right)^{m_c}$ is a symmetric function due to the fact that any two of its arguments can be interchanged without changing the value of the function.
So we now  need to show  that the function  $f(x_1,\cdots,x_n)=\prod\limits_{i=1}^{N}\left(\frac{1}{1+kx_i}\right)^{a} $ is a convex function where $x_i\geq 0$, $a > 0$. Using the Theorem \ref{schur_convex}, it is apparent that if the function $f(x)=\left(\frac{1}{1+kx}\right)^{a} $ is convex function then $f(x_1,\cdots,x_n)=\prod\limits_{i=1}^{N}\left(\frac{1}{1+kx_i}\right)^{a} $ would be convex function. Note that $f(x)=\left(\frac{1}{1+kx}\right)^{a} $ is convex function when $x \geq 0$ and $a>0$ due to the fact that double differentiation of $f(x)=\left(\frac{1}{1+kx}\right)^{a} $ is always non negative, i.e., $f''(x)=a (a+1) k^2 \left(\frac{1}{k x+1}\right)^{a+2} \geq 0$. Thus, $f(x_1,\cdots,x_n)=\prod\limits_{i=1}^{N}\left(\frac{1}{1+kx_i}\right)^{a} $ is a convex function.

Since $\prod\limits_{i=1}^{N}\left(\frac{1}{{1+kx_i}}\right)^{m_c}$ is a convex function and  a symmetric function therefore, it is a Schur-convex function. We have shown that  $\boldsymbol{\hat{\lambda}}\succ \boldsymbol{\lambda}$  and $\prod\limits_{i=1}^{N}\left(\frac{1}{1+kx_i}\right)^{m_c}$ is a Schur-convex function. Therefore, from Proposition 1,  
 $\prod\limits_{i=1}^{N}\left(\frac{1}{{1+k\hat{\lambda}_i}}\right)^{m_c}\geq \prod\limits_{i=1}^{N}\left(\frac{1}{{1+k\lambda_i}}\right)^{m_c}$.
\end{proof}
Thus, the coverage probability in the presence of correlation among the interferers  is greater than or equal  to the coverage probability in  the independent case,  when user channel undergoes Rayleigh fading and the interferers shape parameter $m_i=m_c \text{ }\forall i$. The same result was shown for the case of both Rayleigh user channel and Rayleigh interferers in \cite{7015633} using Vieta formula. By exploiting the mathematical tool of Majorization we are able to provide a much simpler proof.
Next, we compare the coverage probability for general case, i.e., when   $m$ is arbitrary.
\begin{theorem}
\label{theorem_shape}
The coverage probability in the presence of the correlated interferers is greater than   or equal to the coverage probability in presence of independent interferers, when user channel's shape parameter is less than or equal to $1$, i.e., $m \leq 1$. When  $m>1$, coverage probability in the presence of independent is not always lesser than the coverage probability in the presence of correlated interferers. \end{theorem}
\begin{proof}
Please see Appendix.
\end{proof}
Summarizing, the coverage probability in the presence of correlated interferers is greater than or equal to the coverage probability in presence of independent interferers, when user channel's shape parameter is less than or equal to $1$, i.e., $m\leq 1$. When $m>1$, one can not say whether coverage probability is better in correlated interferer case or independent interferer case. Note that when $m\leq 1$, usually the interferers $m_i$ is also smaller than $1$. However, the given proof   holds for both $m_i>1 $ and $m_i<1$.

\subsection{$\eta$-$\mu$ fading }
In this subsection, we analyse the impact of correlation on coverage probability for the scenario where SoI experiences Nakagami-m fading and interferers experience $\eta$-$\mu$ fading.  Recently, the  $\eta$-$\mu$ fading distribution with two shape parameters $\eta$ and $\mu$ has been proposed to model a general non-line-of-sight propagation scenario \cite{4231253}.  It includes Nakagami-q (Hoyt), one sided Gaussian, Rayleigh and Nakagami-m as special cases. The coverage probability expression when SoI experience Nakagami-m fading and interferers experience independent $\eta$-$\mu$ fading with equal $\mu$, i.e.,  $\mu_i=\mu_c \forall i$ is given in \eqref{eq:cov_eta_mu_system}. 
The coverage probability expression when SoI experience Nakagami-m fading and interferers experience correlated $\eta$-$\mu$ fading with equal $\mu$, i.e.,  $\mu_i=\mu_c \forall i$ is given in \eqref{eq:cov_eta_mu_corr_system}.  Note that the functional form of the coverage probability expressions given in \eqref{eq:cov_eta_mu_system} and \eqref{eq:cov_eta_mu_corr_system}  are similar to the case when both SoI and interferers experience Nakagami-m fading.  Hence, the same analysis holds and the coverage probability in the presence of correlated interferers is higher than the coverage probability in the presence of independent interferers when  the user's  channel shape parameter is less than or equal to $1$, and interfering channel experience $\eta$-$\mu$ fading. However, one can not conclude anything when the user's  channel shape parameter is higher than $1$.
\section{Comparison of rate}
In this section, we analyse the impact of correlation on the rate for the scenario when both SoI and interferers experience $\eta$-$\mu$ fading with arbitrary parameters.
In other words, we compare the rate when interferers are independent with the rate when the interferers are correlated using  stochastic ordering theory.

The rate of a user at a distance $r$ when interferers are independent is $R=E[\ln(1+\frac{S}{I})]$, where $S$ is the desired user channel power. For correlated case, the average rate of a user at a distance $r$ is $\hat{R}=E[\ln(1+\frac{S}{\hat{I}})]$.  Here $I$ and $\hat{I}$ are the sum of  independent and  correlated interferers, respectively.  Using iterated expectation one can rewrite the rate as    
\begin{equation}
R=E_S\left[E_I\left[\ln\left(1+\frac{S}{I}\right)\bigg|S=s\right]\right], \text{ and } \hat{R}=E_S\left[E_{\hat{I}}\left[\ln\left(1+\frac{S}{\hat{I}}\right)\bigg|S=s\right]\right].
\end{equation}
Since the  expectation operator preserves inequalities, therefore if we can show that 
\begin{equation*}
E_{\hat{I}}\bigg[\ln\bigg(1+\frac{S}{\hat{I}}\bigg) \bigg|S=s\bigg]\geq E_I\left[\ln\left(1+\frac{S}{I}\right)\bigg|S=s\right],
\end{equation*}
then this implies $\hat{R}\geq R$.

The sum of interference power in the independent case can be written as  
\begin{equation}
I=\sum\limits_{i=1}^{N}h_id_i^{-\alpha}
\end{equation}
where $h_i$ is $\eta$-$\mu$ power RV.  It has been shown in \cite{5288931} that the $\eta$-$\mu$ power RV can be represented as the sum of two gamma RVs with suitable parameters. In other words, if $h_i$ is $\eta$-$\mu$ power RV then, 
\begin{equation}
h_i=x_i+y_i \text{ where } x_i\sim \mathcal{G}\left(\mu_c,\frac{1}{2\mu_c(1+\eta_i^{-1})}\right) \text{ and } y_i\sim \mathcal{G}\left(\mu_c,\frac{1}{2\mu_c(1+\eta_i)}\right)
\end{equation}
Now, the sum of interference power when interference experience $\eta$-$\mu$ RV can be written as
\begin{equation}
I=\sum\limits_{i=1}^{N}h_id_i^{-\alpha}=\sum\limits_{i=1}^{2N}\lambda_iG_i \text{ with } G_i\sim \mathcal{G}(\mu_c,1), \lambda_{2i-1}=\frac{d_i^{-\alpha}}{\mu_c(1+\eta_i^{-1})} \text{ and } \lambda_{2i}=\frac{d_i^{-\alpha}}{\mu_c(1+\eta_i)}.\label{eta_mu1}
\end{equation}
Similarly, for correlated case,
\begin{equation}
\hat{I}=\sum\limits_{i=1}^{N}\hat{h}_i=\sum\limits_{i=1}^{2N} \hat{\lambda}_iG_i \label{eta_mu}
\end{equation}
  Recall that these $\hat{h}_i$  are correlated, and  $\hat{\lambda}_i$s are the eigenvalues of the matrix $\boldsymbol{A}_{\eta}=\boldsymbol{D}_{\eta}\boldsymbol{C}_{\eta}$. In other words one can obtain a correlated sum of gamma variates by multiplying independent and identical distributed (i.i.d.) gamma variates with weight $\hat{\lambda}_i$s. Now, we use Theorem \ref{stochastic_order} (given in Section III) to show that $\hat{R}$ is always greater than equal to $R$. Note that a sequence of RVs $X_1,\cdots X_N$ is said to be exchangeable if for all $N$ and $\pi\in S(N)$ it holds that $X_1\cdots X_N \stackrel{\mathcal{D}}{=} X_{\pi(1)}\cdots X_{\pi(N)}$ where $S(N)$ is the group of permutations of $\{1,\cdots N\}$ and $\stackrel{\mathcal{D}}{=}$ denotes equality in distribution \cite{exch}. Furthermore, if $X_i$s are identically distributed,  they are exchangeable \cite[P. 129]{shaked2007stochastic}. Hence $G_i$s are exchangeable since they are identically distributed. It has already been shown that  $\boldsymbol{\hat{\lambda}} \succ\boldsymbol{ \lambda}$ in Section IV.  Hence by a direct application of  Theorem \ref{stochastic_order} (given in Section III), one obtains, $I \leq_{cx} \hat{I}$.

Note that $\ln(1+\frac{k}{x})$ is a convex function when $k\geq 0$ and $x \geq 0$ due to the fact that double differentiation of $\ln(1+\frac{k}{x})$ is always non negative, i.e., $\frac{\partial \frac{\partial \ln \left(\frac{k}{x}+1\right)}{\partial x}}{\partial x}=\frac{k (k+2 x)}{x^2 (k+x)^2}\geq 0 $. Note that $S$ and $I$ are non negative RVs, hence by a direct application of Theorem \ref{stochastic_order1} (given in Section III), one obtains 
\begin{equation}
E_{\hat{I}}\left[\ln\left(1+\frac{S}{\hat{I}}\right)\bigg|S=s\right]\geq E_I\left[\ln\left(1+\frac{S}{I}\right)\bigg|S=s\right]
\end{equation}
Since expectation preserve inequalities therefore, $E_S[E_I[\ln(1+\frac{S}{I})]] \leq E_S[E_{\hat{I}}[\ln(1+\frac{S}{\hat{I}})]]$. In other words, positive correlation among the interferers increases the rate.

Summarizing, the rate in the presence of the positive correlated interferers is greater than or equal to the rate in the presence of independent interferers, when SoI and interferers both experience $\eta$-$\mu$ fading. Now we briefly discuss the utility of our results in the presence of log normal shadowing.
\subsection{Log Normal Shadowing}
Although all the analysis so far  (comparison of the coverage probability and rate) considered only small scale fading and path loss, the analysis can be further extended to take into account shadowing effects. In general, the large scale fading, i.e, log normal shadowing   is modeled by zero-mean log-normal distribution which is given by,
\begin{equation*}
f_X(x)=\frac{1}{x\sqrt{2\pi(\frac{\sigma_{dB}}{8.686})^2}}\exp\left(-\frac{\ln^2(x)}{2(\frac{\sigma_{dB}}{8.686})^2}\right), x>0,
\end{equation*}
where $\sigma_{dB}$ is the shadow standard deviation represented in dB. Typically the value of $\sigma_{dB}$ varies from $3$ dB to $10$ dB \cite{3gpp},\cite{3gpp1}. It is shown in \cite{generalized} that the pdf of the  composite fading channel (fading and shadowing) can be expressed using the generalized-K (Gamma-Gamma) model. Also in \cite{moment}, it has been shown that the generalized-K pdf can be well approximated by Gamma pdf $\mathcal{G}(\beta, \gamma)$ using the moment matching method, with $\beta$ and $\gamma$ are given by 
\begin{equation}
\beta=\frac{1}{(\frac{1}{m}+1)\exp((\frac{\sigma_{dB}}{8.686})^2)-1}=\frac{m}{({m}+1)\exp((\frac{\sigma_{dB}}{8.686})^2)-m} \label{eq:shadow1} 
\end{equation} 
\begin{equation}  
\text{ and } \gamma=\frac{1+m}{m}\exp\left(\frac{3(\frac{\sigma_{dB}}{8.686})^2}{2}\right)-\exp\left(\frac{(\frac{\sigma_{dB}}{8.686})^2}{2}\right) \label{eq:shadow}
\end{equation} 
Thus, SIR $\eta_l$ of a user can now given by 
\begin{equation}
\eta_l(r)=\frac{Pkr^{-\alpha}}{\sum\limits_{i\in \phi}Pl_{i}d_i^{-\alpha}}
\end{equation}
where $k\sim \mathcal{G}(\beta, \gamma)$ and $l_{i}\sim \mathcal{G}(\beta_i, \gamma_i)$. Here 
\begin{equation}
\beta_i=\frac{m_i}{(m_i+1)\exp((\frac{\sigma_{dB}}{8.686})^2)-m_i} \text{ and }\gamma_i=\frac{(1+m_i)}{m_i}\exp\left(\frac{3(\frac{\sigma_{dB}}{8.686})^2}{2}\right)-\exp\left(\frac{(\frac{\sigma_{dB}}{8.686})^2}{2}\right)
\end{equation}
Further, the correlation coefficient between two identically distributed generalized-K RVs is derived in \cite[Lemma 1]{5501945}, and it is in terms of correlation coefficient  of the RVs corresponding to the short term fading component ($\rho_{i,j}$) and the correlation coefficient of the RVs corresponding to the shadowing component($\rho^s_{i,j}$). The resultant correlation coefficient ($\rho^l_{i,j}$) is then given by
\begin{equation}
\rho^l_{i,j}=\frac{\frac{\rho_{i,j}}{\left( \exp(\frac{\sigma_{dB}}{8.686})^2)-1\right)}+\rho^s_{i,j} m_i+\rho_{i,j}\rho^s_{i,j}}{m_i+\frac{1}{\left( \exp(\frac{\sigma_{dB}}{8.686})^2)-1\right)}+1}\label{rho}
\end{equation}

Note that after approximation, the SIR expression in the presence of log normal shadowing is similar to the SIR expression given in \eqref{sir}, where only small scale fading is present. Hence now  the coverage probability and  rate for the independent case and correlated case can be compared using the methods outlined in Section IV and Section V. In other words, it can be shown that  the coverage probability in the presence of correlated interferers is greater than or equal to the coverage probability in presence of independent interferers, when user's shape parameter is less than or equal to $1$, i.e., $\beta\leq 1$, in the presence of shadow fading. Also, the rate in the presence of positive correlated interferers is always greater than or equal to the rate in the presence of independent interferers, in the presence of shadow fading.

\subsection{Physical interpretation of the impact of correlated Interferers}
We know that the Nakagami-m distribution considers a signal composed of $n$ number of clusters of multipath waves. Within each cluster, the phases of scattered waves are random and have similar delay times. The delay-time spreads of different clusters are relatively large. More importantly, the non-integer Nakagami parameter $m$ is the real extension of integer $n$. One of the primary reason of parameter $m$ being real extension of $n$ is the non zero correlation among the clusters of multipath components\cite{4231253}. In other words, if there exist  correlation among the  clusters, the shape parameter of Nakagami-m fading decreases. 
 
Now, in order to find the physical interpretation of the impact of correlated interferers, we consider a cellular system where the interferes are equidistant and also the shape parameters are identical for every interference. When all the interferers are independent, the shape parameter of the total interference to be $mN$ where the shape parameter of each interfere is $m$ and the total number of interference is $N$ (since the sum of Gamma RV is Gamma RV). However, when there exists  correlation among the interferers, i.e., there exists correlation among the clusters of different interferers, the shape parameter of the total interference  decreases as observed in \cite{4231253}. For example, if the interferers are completely correlated, the shape parameter of total interference is only $m$, whereas it was $mN$, when the interferers were independent. The smaller the shape parameter more faster is the power attenuation of interferers. Hence, the rate increases when there exists correlation among  interferers.

 In the next section, we will show simulation results and discuss how those match with the theoretical results.
\section{Numerical Analysis and Application}
In this section, we study the  impact of correlation among interferers on the coverage probability and rate using simulations and numerical analysis. For simulations, we have considered the classic $19$ cell system associated with a hexagonal structure as shown in Fig. \ref{fig:hexagonal}. For a user we generate the channel fading power corresponding to its own channel as well as that corresponding to the $18$ interferers  and then compute the SIR per user. For correlation scenario, we generate correlated channel fading power  corresponding to the $18$ interferers and then compute SIR per user. Furthermore, based on the SIR, we find coverage probability and rate. Fig. \ref{fig:fig1}  depicts the impact of correlation among the interferers on the coverage probability for different values of shape parameter. Note that only small scale fading is considered in Fig. \ref{fig:fig1}. The  correlation among the interferers is defined by the correlation matrix in \eqref{correlation}  with  $\rho_{pq}=\rho^{|p-q|}$ where $p,q=1,\cdots ,N$ \cite{reig2008performance}. From Fig. \ref{fig:fig1}, it can be observed that for  $m=0.5$ and $m=1$, coverage probability in presence of correlation is higher than that of independent scenario (which match our analytical result).  For example, at $m=0.5$, coverage probability increases from $0.16$ in the independent case to $0.20$ in the correlated case and at $m=1$, coverage probability increases from $0.148$ to $0.216$ when user is at normalized distance $0.7$ from the BS. Whereas $m=3$, one cannot say  that coverage probability in presence of correlation is higher or lower than that of independent scenario. In other words, the coverage probability of independent interferers is higher than the coverage probability of correlated interferers when user is close to the BS. However, the coverage probability of independent interferers is significantly lower than the coverage probability of correlated interferers when the user is far from the BS.

Fig. \ref{fig:correlation} shows the  impact of correlation among the interferers on the coverage probability for different values of correlation coefficient. Here both small scale fading, i.e., $\eta$-$\mu$ fading and large scale fading, i.e., log normal shadowing are considered. The correlation among small scale fading is denoted by $\rho_s$ and the correlation among the large scale fading is denoted by $\rho_l$. It can be seen that as correlation coefficient increases the coverage probability with increases for the correlated interferers case.

\begin{figure}[ht]
 \centering
 \includegraphics[scale=0.4]{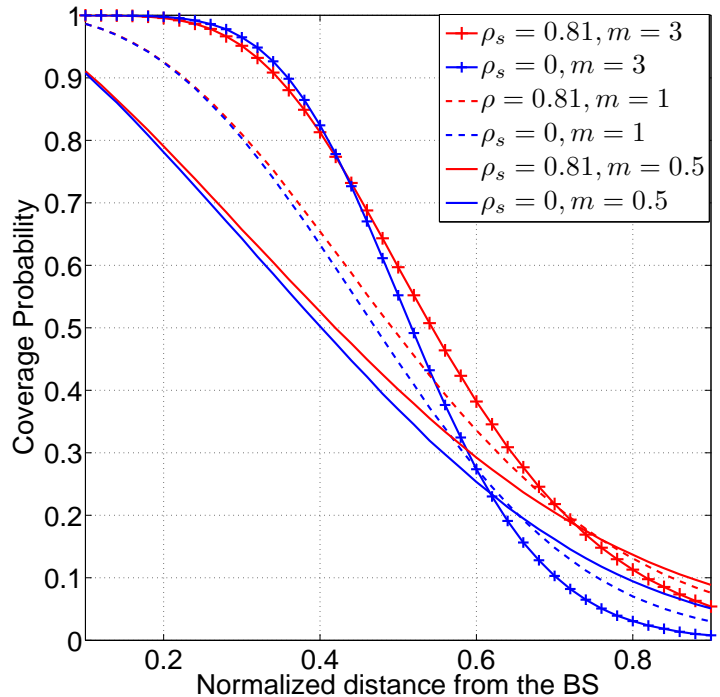}
 \caption{ Coverage probability plot for different value of user's shape parameter, i.e., $m$. Here $ m_c=1,\alpha=2.5, T=3$dB}
 \label{fig:fig1}
\end{figure}

            \begin{figure}[ht]
            \centering
            \includegraphics[scale=0.4]{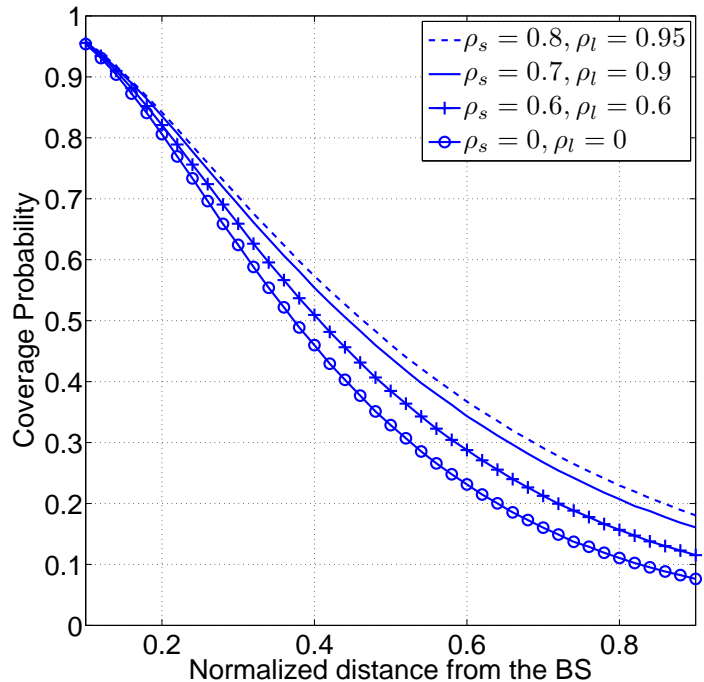}
\caption{ Coverage probability plot for different value of correlation coefficients. Here $m=1,\sigma_{dB}=10dB, \mu_c=1,\eta_i=2,\alpha=2.5, T=3$dB}
             \label{fig:correlation}
             \end{figure}
            \begin{figure}[ht]
            \centering
            \includegraphics[scale=0.42]{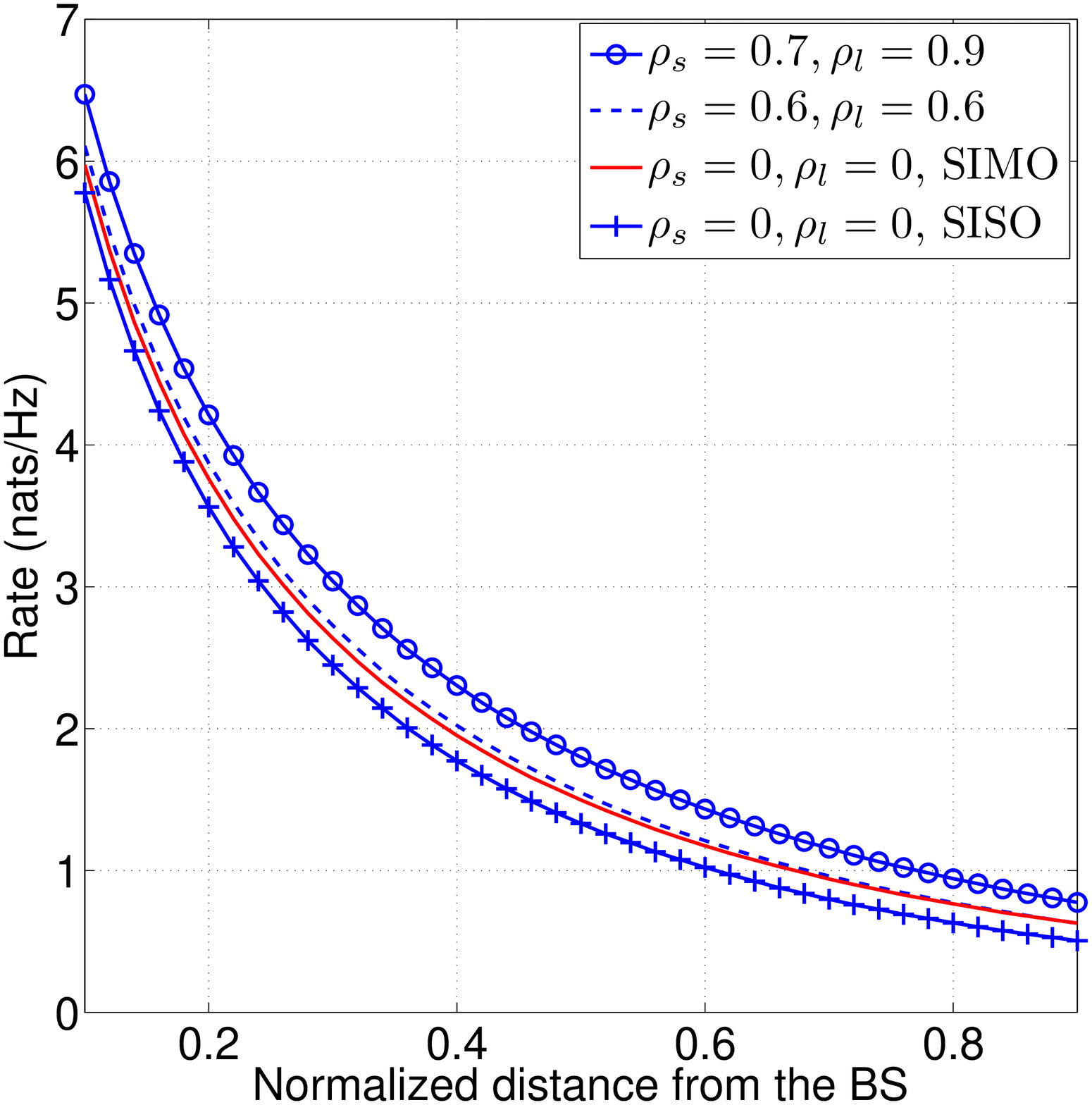}
\caption{rate plot for different value of correlation coefficients. Here $m=1,\sigma_{dB}=10dB, m_i=1,\alpha=2.5$}
             \label{fig:rate}
             \end{figure}
\subsection{How the User can Exploit Correlation among Interferers}
We will now briefly discuss how the user in a cellular network can exploit knowledge of positive correlation among its interferers.
We  compare the coverage probability in the presence of correlated interferers for single input single output (SISO) network with the coverage probability in the presence of independent interferers for single input multiple output (SIMO) network to show that the impact of correlation is significant.  For the SIMO network, it is assumed that each user is equipped with $2$ antennas and both antennas at the user are used for reception since downlink is considered. A linear minimum mean-square-error (LMMSE) receiver \cite{tse2005fundamentals} is considered. In order to calculate rate with a LMMSE receiver, it is assumed that the closest interferer can be completely  cancelled  at the SIMO receiver. Fig. \ref{fig:rate} plots the SISO rate in the presence of independent and correlated interferers case and the rate in the presence of independent interferers for a SIMO network. It can be seen that for $\rho_s=0.6, \rho_l=0.6$, the SISO rate for the correlated case\footnote{The  correlation among the interferers is defined by the correlation matrix in \eqref{correlation}  with  $\rho_{pq}=\rho^{|p-q|}$ where $p,q=1,\cdots ,N$}  is slightly higher than the SIMO rate for independent case. However, for $\rho_s=0.7,\rho_l=0.9$, SISO rate with correlated interferers is significantly higher than the SIMO rate with independent interferers.   In other words, correlation among the interferers seems  to be as good as having one additional antenna at the receiver capable of cancelling the dominant interferer. Obviously, if one had correlated interferers in the SIMO system that would again lead to improved coverage probability and rate and may be compared to a SIMO system with higher number of antennas. In all three cases, it is apparent that if the correlation among the interferers is exploited, it leads to  performance results for a SISO system which are comparable to the   performance of a $1\times 2$ SIMO system with independent interferers.

We have now consider a MU-MIMO system and show that the impact of correlation on MU-MIMO is significant. It is assumed that each user and BS  are equipped with $2$ receive antennas and $2$ transmit antennas, respectively. It is also assume that both transmit antennas at the BS are utilized to transmit $2$ independent data streams to its own $2$ users. A LMMSE receiver is considered and assume that user can cancel the closest interferer. Hence, in this MU-MIMO system, the user will experience no intra-tier interference coming from the serving BS.  Fig. \ref{fig:rate_mumimo} plots the MU-MIMO rate in the presence of independent and correlated interferers case. It can be seen that there is significant gain in rate for correlated case. For example, at normalized distance $0.6$ from the BS, rate increases from $1.24$ nats/Hz in the independent case to $2.55$ nats/Hz in the correlated case when $\rho_s=0.8$ and $\rho_l=0.95$ and $1.985$ nats/Hz in the correlated case when $\rho_s=0.6$ and $\rho_l=0.6$. The reason for significant gain in rate is due to the fact that interferers becomes double for $2\times 2$ MU-MIMO system.

            \begin{figure}[ht]
            \centering
            \includegraphics[scale=0.42]{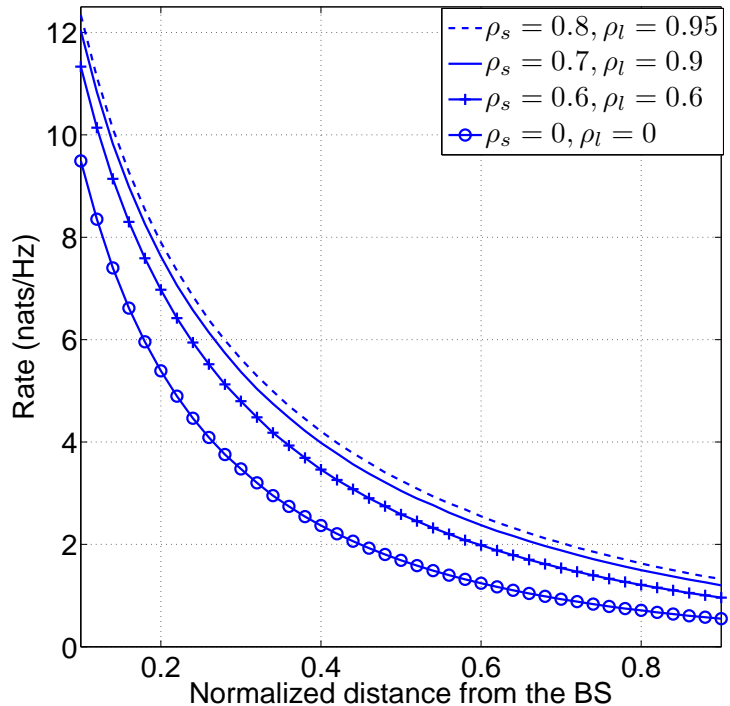}
\caption{rate of MU-MIMO for different value of correlation coefficients. Here $m=1,\sigma_{dB}=10dB, m_i=1,\alpha=2.5$}
             \label{fig:rate_mumimo}
             \end{figure}

We would also likely to briefly point out that the impact of correlation among the interferers is like that of introducing interference alignment in a system. Interference alignment actually aligns interference using appropriate precoding so as reduce the number of interferers one needs to cancel. Here the physical nature of the wireless channel and the presence of co-located interferers also ``aligns'' the interferer partially. This is the reason one can get a gain equivalent to $1\times 2$ system in a $1\times 1$ system with correlated interferers provided the user knows about the correlation. 

Summarizing, our work is able to analytically shows the impact of correlated interferers on coverage probability and rate. This can be used by the network and user to decide whether one wants to use the antennas at the receiver for diversity gain or interference cancellation depending on the information available about interferers correlation. Note that interferers from adjacent sector of a BS will definitely be correlated \cite{correlation,1105, 82767, 104090, 490700, 944855}. We believe that this correlation should be exploited, since the analysis shows that knowledge of correlation will lead to higher coverage probability and rate.  

\section{Conclusions}
In this work, the coverage probability  have been compared analytically for following two cases: $(a)$  User channel experiences Nakagami-m fading and interferers experience $\eta$-$\mu$ fading. $(b)$ Interferers being correlated where the correlation is specified by a correlation matrix. We have  shown that the coverage probability in correlated interferer case is higher than that of the independent case, when the user channel's shape parameter is lesser than or equal to one, and the interferers have Nakagami-m fading with arbitrary parameters. Further, rate have been compared when both user channel and interferers experience $\eta$-$\mu$ fading. It has been shown that positive correlation among the interferers always increases the rate.  We have also taken into account the large scale fading component in our analysis. The impact of correlation seems even more pronounced in the presence of  shadow fading. Moreover, MU-MIMO system is considered and it has been shown that the impact of correlation among interferers is significant on MU-MIMO. Our results indicate that if the user is aware of the interferers correlation matrix then it can exploit it since the correlated interferers behave like partially aligned interferers. This means that if the user is aware of the correlation then one may be able to obtain a rate equivalent to a $1\times 2$ system in a $1\times 1$ system depending on the correlation matrix structure. Extensive simulations  were performed and these match with the theoretical results.

\appendix
\label{app}
\section*{Proof of Theorem \ref{theorem_shape}}
The coverage probability expressions for the scenario when interferers are independent and the scenario when interferers are correlated  are given as follows (given in \eqref{eq:cov9} and \eqref{eq:cov8}, respectively). 
\begin{equation}
C_p(r)=K F_D^{(N)}\left[1-m,m_c,\cdots,m_{c}; Nm_c+1;\frac{1}{Tr^{\alpha}  \lambda_1m+1},\cdots,\frac{1}{Tr^{\alpha}\lambda_N m+1}\right] \label{compare}
\end{equation}
\begin{equation}
\hat{C}_p(r)= \hat{K} F_D^{(N)}\left[1-m,m_c,\cdots,m_{c}; Nm_c+1;\frac{1}{Tr^{\alpha}m\hat{\lambda}_1+1},\cdots,\frac{1}{Tr^{\alpha}m\hat{\lambda}_N+1}\right] \label{compare1}
\end{equation}
where $K=\frac{\Gamma(Nm_c+m)}{\Gamma\left(Nm_c+1\right)}\frac{1}{\Gamma{(m)}} \prod\limits_{i=1}^{N}\left(\frac{1}{Tr^{\alpha}m \lambda_i +1}\right)^{m_c}$ and $\hat{K}=\frac{\Gamma(Nm_c+m)}{\Gamma\left(Nm_c+1\right)}\frac{1}{\Gamma{(m)}} \prod\limits_{i=1}^{N}\left(\frac{1}{Tr^{\alpha}m\hat{\lambda}_i +1}\right)^{m_c}$.  From Theorem \ref{main_thoerem} it is clear that $\hat{K}> K$ . Now, we need to compare the Lauricella's function of the fourth kind of \eqref{compare} and \eqref{compare1}. Here, for comparison we use the series expression for $F_D(.)$.
We expand the  series expression for the Lauricella's function of the fourth kind in the  following form:
\begin{align}
F_D^{(N)}[a,b,\cdots, b;c;x_1,\cdots, x_N]=&1+K_{1,1}\sum\limits_{i=1}^{N}x_i+K_{2,1}\sum\limits_{i=1}^{N}x_i^2+K_{2,2}\sum\limits_{1\leq i<j\leq N}x_ix_j+K_{3,1}\sum\limits_{i=1}^{N}x_i^3\nonumber\\
&+K_{3,2}\sum\limits_{i,j=1, s.t.i\neq j}^{N}x_i^2x_j+K_{3,3}\sum\limits_{1\leq i<j<k\leq N}x_ix_jx_k+\cdots\label{lauricella}
\end{align}
 where
$K_{1,1}=\frac{(a)_1(b)_1}{(c)_1 1!}$, $K_{2,1}=\frac{(a)_2(b)_2}{(c)_2 2!}$, $K_{2,2}=\frac{(a)_2(b)_1(b)_1}{(c)_21!1!}$, $K_{3,1}=\frac{(a)_3(b)_3}{(c)_3 3!}$, $K_{3,2}=\frac{(a)_3(b)_2(b)_1}{(c)_3 2!1!}$, $K_{3,3}=\frac{(a)_3(b)_1(b)_1(b)_1}{(c)_31!1!1!}$ and so on.

Hence the coverage probability for independent case given in \eqref{compare} can be written as
\begin{equation*}
\textstyle C_p(r)=K\bigg[ 1+K_{1,1}\sum\limits_{i=1}^{N}\left(\frac{1}{Tr^{\alpha}m \lambda_i +1}\right)+K_{2,1}\sum\limits_{i=1}^{N}\left(\frac{1}{Tr^{\alpha}m \lambda_i +1}\right)^2+ K_{3,1}\sum\limits_{i=1}^{N}\left(\frac{1}{Tr^{\alpha}m \lambda_i +1}\right)^3+
\end{equation*}
\begin{equation*}
\textstyle
K_{2,2}\sum\limits_{1\leq i<j\leq N}\left(\frac{1}{Tr^{\alpha}m \lambda_i +1}\right)\left(\frac{1}{Tr^{\alpha}m \lambda_j +1}\right)+
K_{3,2}\sum\limits_{i,j=1, s.t. i\neq j}^{N}\left(\frac{1}{Tr^{\alpha}m \lambda_i +1}\right)^2\left(\frac{1}{Tr^{\alpha}m \lambda_j +1}\right)+ 
\end{equation*}
\begin{equation}
\textstyle
K_{3,3}\sum\limits_{1\leq i<j<k\leq N}\left(\frac{1}{Tr^{\alpha}m \lambda_i +1}\right)\left(\frac{1}{Tr^{\alpha}m \lambda_j +1}\right)\left(\frac{1}{Tr^{\alpha}m \lambda_k +1}\right)+\cdots\bigg]\label{sum1}
\end{equation}
Similarly,  for the correlated  case  the coverage probability given in \eqref{compare1} can be written as
\begin{equation*}
\textstyle
\hat{C}_p(r)=\hat{K}\bigg[ 1+K_{1,1}\sum\limits_{i=1}^{N}\left(\frac{1}{Tr^{\alpha}m\hat{\lambda}_i +1}\right)+K_{2,1}\sum\limits_{i=1}^{N}\left(\frac{1}{Tr^{\alpha}m\hat{\lambda}_i +1}\right)^2+ K_{3,1}\sum\limits_{i=1}^{N}\left(\frac{1}{Tr^{\alpha}m\hat{\lambda}_i +1}\right)^3+
\end{equation*}
\begin{equation*}
\textstyle
K_{2,2}\sum\limits_{1\leq i<j\leq N}\left(\frac{1}{Tr^{\alpha}m\hat{\lambda}_i +1}\right)\left(\frac{1}{Tr^{\alpha}m\hat{\lambda}_j +1}\right)+
K_{3,2}\sum\limits_{i,j=1, s.t. i\neq j}^{N}\left(\frac{1}{Tr^{\alpha}m\hat{\lambda}_i +1}\right)^2\left(\frac{1}{Tr^{\alpha}m\hat{\lambda}_j +1}\right)+
\end{equation*}
\begin{equation}
\textstyle
K_{3,3}\sum\limits_{1\leq i<j<k\leq N}\left(\frac{1}{Tr^{\alpha}m\hat{\lambda}_i +1}\right)\left(\frac{1}{Tr^{\alpha}m\hat{\lambda}_j +1}\right)\left(\frac{1}{Tr^{\alpha}m\hat{\lambda}_k +1}\right)+\cdots\bigg]\label{sum2}
\end{equation}
Here 
$K_{1,1}=\frac{(1-m)_1(m_c)_1}{(N m_c+1)_1 1!}$, $K_{2,1}=\frac{(1-m)_2(m_c)_2}{(N m_c+1)_2 2!}$, $K_{2,2}=\frac{(1-m)_2(m_c)_1(m_c)_1}{(N m_c+1)_2 1!1!}$, $K_{3,1}=\frac{(1-m)_3(m_c)_3}{(N m_c+1)_3 3!}$, $K_{3,2}=\frac{(1-m)_3(m_c)_2(m_c)_1}{(N m_c+1)_3 2!1!}$, $K_{3,3}=\frac{(1-m)_3(m_c)_1(m_c)_1(m_c)_1}{(N m_c+1)_31!1!1!}$ and so on. Note that here $K_{i,j}$ are the same for both $C_p(r)$ and $\hat{C}_p(r)$. Now, we want to show that each summation term in the series expression is a Schur-convex function.

Each summation term in the series expression is symmetrical due to the fact that any two of its argument can be interchanged without changing the value of the function. We have already shown that $\prod\limits_{i=1}^{N}\left(\frac{1}{1+kx_i}\right)^{a} $ is a convex function $\forall x_i\geq 0$ and $\forall a>0$. Now, the terms in the summation terms in \eqref{sum1} and \eqref{sum2} are of the form $\prod\limits_{i=1}^{M}\left(\frac{1}{1+kx_i}\right)^{a_i} $ where $M\leq N$.  Using Theorem \ref{schur_convex} (given in Section II), one can show that each  term of each summation term is a convex function. Using the fact that convexity is preserved under summation one can show that each summation term is a convex function. Thus, each summation term 
in 
series expression is a Schur-convex function.

Now we consider following two cases.

Case I when $m< 1$:
Since $m< 1$, so $1-m > 0$ and hence all the constant $K_{i, j}> 0 \text{ } \forall\text{ } i, j$.  Each summation term in series expression of coverage probability for correlated case is greater  than or equal to the corresponding summation term in the series expression of coverage probability for independent case. Thus, if user  channel's shape parameter $m< 1$ then coverage probability of  correlated case is  greater than or equal to  the coverage probability for independent case.

Case II when $m> 1$:
Since $m > 1$, then $1-m < 0$ and hence $K_{i,j}<0 \text{ }\forall i \in 2|\mathbb{Z}|+1 \text { and } \forall j $  where set $\mathbb{Z}$ denote the integer number, due to  the fact that $(a)_N<0 \text{ if } a<0 \text{ and } N \in 2|\mathbb{Z}|+1 $.  Whereas, $K_{i,j}>0 \text{ }\forall i \in 2|\mathbb{Z}| \text { and } \forall j $ due of the fact that $(a)_N>0 \text{ if } a<0 \text{ and } N\in 2|\mathbb{Z}|$. Thus, if   $m>1$, we cannot state whether the coverage probability of one case is greater than or lower than the other case.

\bibliographystyle{IEEEtran}
\bibliography{bibfile}

\end{document}